\documentclass[proceedings]{stacs}
\stacsheading{2010}{405-416}{Nancy, France}
\firstpageno{405}

\usepackage{times}
\usepackage{amsmath,amssymb,eufrak}
\usepackage{enumerate}
\usepackage{complexity}
\usepackage{gastex}

\usepackage[active]{srcltx}

\def\AC{{\mathsf{AC}}}

\def\BIN{{\mathrm{BIN}}}
\def\CRR{{\mathrm{CRR}}}
\def\CTL{{\mathsf{CTL}}}
\def\div{{\mathrm{div}}}

\def\EF{{\mathsf{EF}}}
\def\EXPTIME{{\mathsf{EXPTIME}}}
\def\inp{{\mathsf{in}}}

\def\NC{{\mathsf{NC}}}
\def\Omc{{\mathbb{O}}}

\def\out{{\mathsf{out}}}
\def\Pmc{{\mathcal{P}}}
\def\path{{\mathrm{path}}}
\def\PSPACE{{\mathsf{PSPACE}}}
\def\TC{{\mathsf{TC}}}
\def\U{{\mathsf{U}}}

\newcommand{\ValOne}{\text{ValOne}}
\newcommand{\OptValOne}{\text{OptValOne}}
\newcommand{\LCM}{\text{LCM}}

\theoremstyle{plain}
\theoremstyle{plain}
\theoremstyle{definition}

\begin{document}

\title[Branching-time model checking of one-counter processes]{Branching-time model checking\\ of one-counter processes}

\author[lab1]{S. G\"oller}{Stefan G\"oller}
\address[lab1]{
Universit\"at Bremen, Fachbereich Mathematik und Informatik} 
\email{goeller@informatik.uni-bremen.de} 

\author[lab2]{M. Lohrey}{Markus Lohrey}
\address[lab2]{Universit\"at Leipzig, Institut f\"ur Informatik}
\email{lohrey@informatik.uni-leipzig.de}

\thanks{The second author
would like to acknowledge the support by DFG research project GELO} 


\begin{abstract}
\noindent 
One-counter processes (OCPs) are pushdown processes which operate only on a
unary stack alphabet. We study the computational complexity of model checking
computation tree logic ($\CTL$) over OCPs. A $\PSPACE$ upper bound is inherited from
the modal $\mu$-calculus for this problem. First, we analyze the periodic behaviour
of $\CTL$ over OCPs and derive a model checking algorithm whose running time is
exponential only in the number of control locations and a syntactic notion of
the formula that we call leftward until depth. Thus, model checking fixed OCPs
against $\CTL$ formulas with a fixed leftward until depth is in $\P$. This generalizes
a result of the first author, Mayr, and To for the expression complexity of
$\CTL$'s fragment $\EF$. Second, we prove that already over some fixed OCP, $\CTL$ model
checking is $\PSPACE$-hard. Third, we show that there already exists a fixed $\CTL$
formula for which model checking of OCPs is $\PSPACE$-hard. For the latter,
we employ two results from complexity theory: (i) Converting a natural
number in Chinese remainder presentation into binary presentation is in
logspace-uniform $\NC^1$ and (ii) $\PSPACE$ is $\AC^0$-serializable. We demonstrate
that our approach can be used to answer further open questions.
\end{abstract}

\maketitle

\section{Introduction}

Pushdown automata (PDAs) (or recursive state machines) are a natural model
for sequential programs with recursive procedure calls, and their verification
problems have been studied extensively.
The complexity of model checking problems for PDAs is quite well understood: 
The reachability problem for PDAs can be solved in  polynomial time
\cite{BoEsMa97,EsHaRoSch00}. Model checking
modal $\mu$-calculus over PDAs was shown to be $\EXPTIME$-complete in
\cite{Wal01}, and the global version of the model checking
problem has been considered in \cite{Cachat02,PitVar04,Serr03}.
The $\EXPTIME$ lower bound
for model checking PDAs also holds for the simpler logic $\CTL$
and its fragment $\mathsf{EG}$ \cite{WalCTL00}, even for a fixed formula
(data complexity) \cite{Boz07} or a fixed PDA (expression complexity).
On the other hand, model checking PDAs against the logic $\EF$
(another natural fragment of $\CTL$) is $\PSPACE$-complete
\cite{WalCTL00}, and again the lower bound still
holds if either the formula or the PDA is fixed \cite{BoEsMa97}.
Model checking problems for various fragments and extensions
of PDL (Propositional Dynamic Logic) over PDAs were studied
in \cite{GolLoh06}.

One-counter processes (OCPs) are Minsky counter machines with just one
counter. They can also be seen as a special
case of PDAs with just one stack symbol, plus a non-removable
bottom symbol which indicates an empty stack (and thus allows to
test the counter for zero) and hence constitute a natural and fundamental
computational model. In recent years, 
model checking problems for OCPs received increasing attention
\cite{GoMaTo09,HKOW-09concur,Serr06,To09}.
Clearly, all upper complexity bounds carry over from PDAs.
The question, whether these upper bounds can be matched by lower bounds
was just recently solved for several important logics:
Model checking modal $\mu$-calculus over OCPs is $\PSPACE$-complete. 
The $\PSPACE$ upper bound was
shown in \cite{Serr06}, and a matching lower bound can
easily be shown by a reduction from emptiness of alternating
unary finite automata, which was shown to be $\PSPACE$-complete in
\cite{Ho96,JaSa07}.
This lower bound even holds if either the OCP or the
formula is fixed.  The situation becomes different for the fragment $\EF$.
In \cite{GoMaTo09}, it was shown that  model checking $\EF$ over 
OCPs is in the complexity class $\mathsf{P}^{\mathsf{NP}}$ 
(the class of all problems that can be solved on a deterministic
polynomial time machine with access to an oracle from $\mathsf{NP}$).
Moreover, if the input formula is represented succinctly
as a directed acyclic graph, then model checking $\EF$ over 
OCPs is also hard for $\mathsf{P}^{\mathsf{NP}}$.
For  the standard (and less succinct) tree representation for formulas,
only hardness for the class $\mathsf{P}^{\mathsf{NP}[\log]}$
(the class of all problems that can be solved on a deterministic
polynomial time machine which is allowed to make $O(\log(n))$ 
many queries to an oracle from $\mathsf{NP}$) was shown in \cite{GoMaTo09}.
In fact, there already exists a fixed $\EF$ formula
such that model checking this formula over a given OCP is hard for
$\mathsf{P}^{\mathsf{NP}[\log]}$, i.e., the data complexity is 
$\mathsf{P}^{\mathsf{NP}[\log]}$-hard.

In this paper we consider the model checking problem for $\CTL$
over OCPs. 
By the known upper bound for the modal $\mu$-calculus
\cite{Serr06} this problem belongs to $\PSPACE$.
First, we analyze the combinatorics of $\CTL$ model checking over OCPs.
More precisely, we analyze the periodic behaviour of the
set of natural numbers that satisfy a given $\CTL$ formula in a given
control location of the OCP (Thm.~\ref{thm-CTL-periodic}).
By making use of Thm.~\ref{thm-CTL-periodic}, we can derive a model checking
algorithm whose running time is exponential
only in the number of control locations and a syntactic measure on 
$\CTL$ formulas that we call leftward until depth
(Thm.~\ref{CTL upper bound_0}).
As a corollary, we obtain that model checking a fixed OCP
against $\CTL$ formulas of fixed leftward until depth lies in $\P$.
This generalizes a recent result
from \cite{GoMaTo09}, where it was shown that the expression complexity of
$\EF$ over OCPs lies in $\P$.
Next, we focus on lower bounds.
We show that model checking $\CTL$ over OCPs is 
$\PSPACE$-complete, even if we fix either 
the OCP (Thm.~\ref{Theo-CTL-expression}) or the $\CTL$ formula (Thm.~\ref{theo ctl data}).
The proof of Thm.~\ref{Theo-CTL-expression} uses a reduction from 
QBF. 
We have to construct a fixed OCP for which we can construct for a given unary
encoded number $i$ $\CTL$ formulas
that express, when interpreted over our fixed OCP, whether the current
counter value is divisible by $2^i$ and whether the $i^{\text{th}}$ bit
in the binary representation of the current counter value is $1$, respectively.  
For the proof of Thm.~\ref{theo ctl data}  ($\PSPACE$-hardness of
data complexity for $\CTL$)
we use two techniques from complexity theory, which to our knowledge have not
been applied in the context of verification so far:
(i) the existence of small
depth circuits for converting a number from Chinese remainder representation
to binary representation  
and
(ii) the fact that $\PSPACE$-computations are serializable in a certain sense
(see Sec.~\ref{S Tools} for details).
One of the main obstructions 
in getting lower bounds for OCPs is the fact that OCPs are well suited
for testing divisibility properties of the counter value and hence can
deal with numbers in Chinese remainder representation, but it is not 
clear how to deal with numbers in binary representation. Small depth circuits
for converting a number from Chinese remainder representation
to binary representation are the key in order to overcome this obstruction.

We are confident that our new lower bound techniques described above can be used for
proving further lower bounds for OCPs. We present two other applications
of our techniques in Sec.~\ref{further applications}: (i)
We show that model checking $\EF$ over OCPs is complete for $\mathsf{P}^{\mathsf{NP}}$ even if the input formula
is represented by a tree (Thm.~\ref{T EF}) and thereby solve an open problem
from \cite{GoMaTo09}.
(ii) We improve a lower bound on a decision problem for one-counter Markov
decision processes from \cite{BraBroEteKucWoj09} (Thm.~\ref{thm markov}).
The following table summarizes the picture on the complexity of model
checking for PDAs and OCPs. Our new results are marked with (*).

\medskip

\begin{center}
\begin{tabular}{l|l|l}
\textbf{Logic}        &  \textbf{PDA}     &  \textbf{OCP}\\ \hline
modal $\mu$-calculus  &  $\EXPTIME$-complete &  $\PSPACE$-complete \\ \hline
modal $\mu$-calculus, fixed formula &  $\EXPTIME$-complete &  $\PSPACE$-complete \\ \hline
modal $\mu$-calculus, fixed system &  $\EXPTIME$-complete &  $\PSPACE$-complete \\ \hline
$\CTL$, fixed formula & $\EXPTIME$-complete  &  $\PSPACE$-complete (*)\\ \hline
$\CTL$, fixed system & $\EXPTIME$-complete  &  $\PSPACE$-complete (*)\\ \hline
$\CTL$, fixed system, fixed leftward until depth & $\EXPTIME$-complete  &  in $\P$ (*)\\ \hline
$\EF$           & $\PSPACE$-complete  &  $\mathsf{P}^{\mathsf{NP}}$-complete (*)\\ \hline
$\EF$, fixed formula & $\PSPACE$-complete & $\mathsf{P}^{\mathsf{NP}[\log]}$-hard, in
$\mathsf{P}^{\mathsf{NP}}$ \\ \hline
$\EF$, fixed system & $\PSPACE$-complete & in $\P$ 
\end{tabular}
\end{center}

\medskip
\noindent
Missing proofs due to space restrictions can be found in the full
version of this paper \cite{GoLo09}.

\section{Preliminaries}{\label{S Pre}}

\newcommand{\N}{\mathbb{N}}
\newcommand{\Z}{\mathbb{Z}}
\renewcommand{\O}{\mathbb{O}}
\newcommand{\bit}{\text{bit}}

We denote the naturals by $\N=\{0,1,2,\ldots\}$.
For $i,j\in\N$ let $[i,j]=\{k\in\N\mid
i\leq k\leq j\}$ and $[j]=[1,j]$. In particular $[0] = \emptyset$.
For $n\in\N$ and $i\geq 1$, let $\bit_i(n)$ denote the $i^{\text{th}}$
least significant bit of the binary representation of $n$, i.e.,
$n=\sum_{i\geq 1} 2^{i-1}\cdot\bit_i(n)$.
For every finite and non-empty subset $M\subseteq\N\setminus\{0\}$, 
define $\LCM(M)$ to be the {\em least common multiple} of all numbers in $M$. 
It is known that  $2^k\leq\LCM([k])\leq4^k$ for all $k\geq 9$ \cite{Nai82}.
As usual, for a possibly infinite alphabet $A$, $A^*$ (resp. $A^\omega$)
denotes the set of all  finite (resp. infinite) words over $A$.
Let $A^\infty = A^* \cup A^\omega$ and $A^+ = A^* \setminus \{\varepsilon\}$,
where $\varepsilon$ is the empty word.  The length of a finite word 
$w$ is denoted by $|w|$.
For a word $w=a_1a_2\cdots a_n \in A^*$ (resp. $w=a_1a_2\cdots \in A^\omega$)
with $a_i\in A$ and $i \in [n]$ (resp. $i\geq 1$),
we denote by  
$w_i$ the $i^{\text{th}}$ letter $a_i$. 
A nondeterministic finite automaton (NFA) is a tuple
$A=(S,\Sigma,\delta,s_0,S_f)$, where $S$ is a finite set of {\em states},
$\Sigma$ is a {\em finite alphabet}, $\delta\subseteq S\times\Sigma\times S$ is the
{\em transition relation}, $s_0\in S$ is the {\em initial state}, and
$S_f\subseteq S$ is 
a set of {\em final states}.
We assume some basic knowledge in complexity
theory, see e.g. \cite{AroBar09} for more details. 

\section{One-counter processes and computation tree logic}{\label{S OCP CTL}}
\newcommand{\Prop}{\mathcal{P}}

Fix a countable set $\Prop$ of {\em propositions}.
A {\em transition system} is a triple $T=(S,\{S_p\mid p\in\Prop\}, \rightarrow)$, where
$S$ is the set of {\em states},
$\to\, \subseteq S \times S$ is the set of {\em transitions}
and $S_p\subseteq S$ for all $p \in \Prop$ with $S_p = \emptyset$ for all but 
finitely many $p\in\Prop$.  We write $s_1\rightarrow s_2$ instead of 
$(s_1,s_2)\in\,\rightarrow$. 
The set of all {\em finite} (resp. {\em infinite}) {\em paths} in $T$
is $\path_+(T) = \{ \pi \in S^+ \mid \forall i \in [|\pi|-1] : \pi_i
\to \pi_{i+1} \}$
(resp. $\path_\omega(T) = \{ \pi \in S^\omega \mid \forall i \geq 1 : \pi_i \to \pi_{i+1} \}$).
For a subset $U\subseteq S$ of states, a (finite or infinite) path $\pi$ is called a {\em $U$-path}
if $\pi \in U^\infty$.

A {\em one-counter process} (OCP) is a tuple
$\O=(Q,\{Q_p\mid p\in\Prop\},\delta_0,\delta_{>0})$, where 
$Q$ is a finite set of {\em control locations}, 
$Q_p\subseteq Q$ for all $p\in\Prop$ with $Q_p = \emptyset$ for all but
finitely many $p\in\Prop$, 
$\delta_0\subseteq Q \times\{0,1\}\times Q$ is a set of
{\em zero transitions},  and 
$\delta_{>0}\subseteq Q\times \{-1,0,1\}\times Q$ is a set of
{\em positive transitions}. 
The {\em size} of the OCP $\O$ is
$|\O|=|Q|+\sum_{p\in\Prop}|Q_p| +|\delta_0|+|\delta_{>0}|$.
The transition system defined by $\O$ is
$T(\O)=(Q\times\N,\{Q_p\times\N\mid p\in\Prop\},\rightarrow)$, where 
$(q,n)\rightarrow(q',n+k)$ if and only if either 
$n = 0$ and $(q,k,q')\in\delta_0$, or
$n > 0$ and $(q,k,q')\in\delta_{>0}$.
A {\em one-counter net} (OCN) is an OCP, where
$\delta_0\subseteq\delta_{>0}$.
For $(q,k,q') \in \delta_0 \cup \delta_{>0}$ we usually
write $q \xrightarrow{k} q'$.

\newcommand{\X}{\mathsf{X}}
\newcommand{\F}{\mathsf{F}}
\renewcommand{\G}{\mathsf{G}}
\newcommand{\WU}{\mathsf{WU}}
\newcommand{\links}{[\![}
\newcommand{\rechts}{]\!]}
\newcommand{\sem}[1]{\ensuremath{\links #1 \rechts}}
\newcommand{\true}{\texttt{true}}
\newcommand{\false}{\texttt{false}}

More details on the temporal logic $\CTL$ can be found 
for instance in \cite{BaiKat08}.
{\em Formulas} $\varphi$ of $\CTL$ are defined by the following
grammar, where $p\in\Prop$:
$$
\varphi\quad::=\quad p\ \mid\ \neg\varphi\ \mid\ \varphi\wedge\varphi\ \mid\ 
\exists\X\varphi\ \mid\ \exists\varphi\U\varphi\ \mid\ \exists\varphi\WU\varphi.
$$ 
Given a transition system $T=(S,\{S_p\mid p\in\Prop\},\rightarrow)$ 
and a $\CTL$ formula $\varphi$, we define 
the semantics $\links\varphi\rechts_T\subseteq S$ 
by induction on the structure of $\varphi$ as
follows:
$\sem{p}_T =S_p \text{ for each } p\in\Prop$,
$\sem{\neg\varphi}_T =  S\setminus\sem{\varphi}_T$,
$\sem{\varphi_1\wedge\varphi_2}_T= 
\sem{\varphi_1}_T\cap\sem{\varphi_2}_T$,
$\sem{\exists\X\varphi}_T = \{s\in S\mid \exists s'\in\sem{\varphi}_T: s\rightarrow
s'\}$, 
$\sem{\exists\varphi_1\U\varphi_2}_T  = 
\{s\in S\mid \exists  \pi \in \path_+(T) : \pi_1 = s, 
\pi_{|\pi|} \in\sem{\varphi_2}_T, \forall i\in[|\pi|-1] :
\pi_i\in\sem{\varphi_1}_T\}$,
$\sem{\exists\varphi_1\WU\varphi_2}_T =  \sem{\exists\varphi_1\U\varphi_2}_T\cup
\{s\in S\mid \exists \pi \in \path_\omega(T) : \pi_1 = s, 
\forall i\geq 1 : \pi_i\in\sem{\varphi_1}_T\}$.
We also write $(T,s)\models\varphi$ (or briefly $s \models\varphi$ 
if $T$ is clear from the context) for 
$s\in\sem{\varphi}_T$. We introduce the usual abbreviations 
$\varphi_1\vee\varphi_2=\neg(\neg\varphi_1\wedge\neg\varphi_2)$,
$\forall \X\varphi=\neg\exists\X\neg\varphi$,
$\exists\F\varphi=\exists (p \vee \neg p)\U\varphi$, and
$\exists\G\varphi=\exists\varphi\WU (p \wedge \neg p)$ for some $p\in\Prop$.
Formulas of the $\CTL$-fragment $\EF$ are given by the following grammar, where
$p\in\Prop$:
$\varphi::= p\ \mid \neg\varphi\ \mid\ \varphi\wedge\varphi\ \mid\
\exists\X\varphi\ \mid \exists\F\varphi$.
The {\em size} of $\CTL$ formulas is defined as follows: $|p|=1$,
$|\neg\varphi|=|\exists\X\varphi|=|\varphi|+1$, $|\varphi_1\wedge\varphi_2|=|\varphi_1|+|\varphi_2|+1$,
$|\exists\varphi_1\U\varphi_2|=|\exists\varphi_1\W\U\varphi_2|=|\varphi_1|+|\varphi_2|+1$.

\section{CTL on OCPs: Periodic behaviour  and upper bounds}\label{S Upper}

\newcommand{\lud}{\mathrm{lud}}
The goal of this section is to prove a periodicity property of $\CTL$ over
OCPs, which implies an upper bound for $\CTL$ on OCPs, see 
Thm.~\ref{CTL upper bound_0}. As a corollary, we state
that for a fixed OCP, $\CTL$ model checking restricted
to formulas of fixed leftward until depth (see the definition below)
can be done in polynomial time.
We define the {\em leftward until depth $\lud$} of $\CTL$ formulas
inductively as follows: $\lud(p)=0$  for $p\in\Prop$,
$\lud(\neg\varphi)= \lud(\exists\X\varphi) = \lud(\varphi)$,
$\lud(\varphi_1\wedge\varphi_2)=\max\{\lud(\varphi_1),\lud(\varphi_2)\}$,
$\lud(\exists\varphi_1\U\varphi_2) = \lud(\exists\varphi_1\W\U\varphi_2) = 
\max\{\lud(\varphi_1)+1,\lud(\varphi_2)\}$.
A similar definition of until depth can be found in \cite{TheWi96}, but
there  the until depth of $\exists\varphi_1\U\varphi_2$ is 
1 plus the maximum of the until depths of $\varphi_1$ and 
$\varphi_2$. Note that $\lud(\varphi)\leq 1$ for every $\EF$ formula $\varphi$.

Let us fix an OCP $\O=(Q,\{Q_p\mid p\in\Prop\},\delta_0,\delta_{>0})$ 
for the rest of this section. 
Let $|Q|=k$ and define $K=\LCM([k])$ and
$K_\varphi=K^{\lud(\varphi)}$ for each $\CTL$ formula $\varphi$.

\begin{theorem} \label{thm-CTL-periodic}
For all $\CTL$ formulas $\varphi$,
all $q \in Q$ and all
$n,n'>2\cdot|\varphi|\cdot k^2\cdot K_\varphi$ with $n\equiv n' \text{ mod } K_\varphi$:
\begin{eqnarray}{\label{E period}}
(q,n)\in\sem{\varphi}_{T(\O)}\quad \Longleftrightarrow \quad 
(q,n')\in\sem{\varphi}_{T(\O)}.
\end{eqnarray} 
\end{theorem}

\begin{proof}[Proof sketch]
We prove the theorem by induction on the structure of $\varphi$. 
We only treat the difficult case $\varphi=\exists\psi_1\U\psi_2$ here.
Let $T = \max\{ 2\cdot|\psi_i|\cdot k^2\cdot K_{\psi_i}\mid i\in\{1,2\}\}$.
Let us prove equivalence (\ref{E period}). 
Note that $K_\varphi=\LCM\{K\cdot K_{\psi_1},K_{\psi_2}\}$ by definition.
Let us fix an arbitrary control location $q\in Q$ and naturals $n,n'\in\N$
such that $2\cdot|\varphi|\cdot k^2\cdot K_\varphi<n<n'$ and $n\equiv n' \text{ mod } K_\varphi$. We have
to prove that  $(q,n)\in\sem{\varphi}_{T(\O)}$ if and only if 
$(q,n')\in\sem{\varphi}_{T(\O)}$.
For this, let $d=n'-n$, which is a multiple of $K_\varphi$.
We only treat the ``if''-direction here and recommend the reader to consult
\cite{GoLo09} for helpful illustrations.
So let us assume that $(q,n')\in\sem{\varphi}_{T(\O)}$.  To prove that
$(q,n)\in\sem{\varphi}_{T(\O)}$, we will use the following claim.

\medskip

\noindent
{\em Claim:}
Assume some $\sem{\psi_1}_{T(\O)}$-path
$\pi = [(q_1,n_1) \to (q_2,n_2) \to \cdots \to (q_l,n_l)]$ 
with $n_i > T$ for all $i \in [l]$ and $n_1-n_l\geq k^2\cdot K\cdot K_{\psi_1}$. 
Then there exists a $\sem{\psi_1}_{T(\O)}$-path
from $(q_1,n_1)$ to $(q_l,n_l+K\cdot K_{\psi_1})$, whose counter values are all
strictly above $T+K\cdot K_{\psi_1}$.

\medskip
\noindent
The claim tells us that paths that lose height
at least $k^2\cdot K\cdot K_{\psi_1}$ and
whose states all have counter values strictly
above $T$ can be flattened (without changing the starting state) 
by height $K\cdot K_{\psi_1}$.

\medskip
\noindent
{\em Proof of the claim.}
For each counter value $h\in\{n_i\mid i\in[l]\}$ that appears in $\pi$, let
$\mu(h)=\min\{i\in[l]\mid n_i=h\}$ denote the minimal position
in $\pi$ whose corresponding state has counter value $h$.
Define $\Delta=k\cdot K_{\psi_1}$. 
We will be interested in
$k\cdot K$ many consecutive intervals (of counter values) each of size $\Delta$.
Define the bottom $b=n_1-(k\cdot K)\cdot\Delta$. 
Formally, an {\em interval} is a set
$I_i=[b+(i-1)\cdot\Delta,b+i\cdot\Delta]$ for some $i\in[k\cdot K]$. 
Since each interval has size $\Delta=k\cdot K_{\psi_1}$, we can think of each
interval
$I_i$ to consist of $k$ consecutive {\em sub-intervals} of size $K_{\psi_1}$ each.
Note that each sub-interval has two extremal elements, namely its {\em upper} and
{\em lower boundary}. Thus all $k$ sub-intervals have $k+1$ boundaries in total.
Hence, by the pigeonhole principle, 
for each interval $I_i$, there exists some $c_i\in[k]$ and 
two distinct boundaries $\beta(i,1) > \beta(i,2)$ of distance $c_i\cdot
K_{\psi_1}$ such that the control location of $\pi$'s earliest state
of counter value $\beta(i,1)$ agrees with the control location of $\pi$'s
earliest state of counter value $\beta(i,2)$, i.e., formally
$q_{\mu(\beta(i,1))}\ =\ q_{\mu(\beta(i,2))}$.
%The situation is depicted in Fig.~\ref{F Blocks}. 
Observe that flattening the path $\pi$ by gluing together $\pi$'s states
at position $\mu(\beta(i,1))$ and $\mu(\beta(i,2))$ (for this,
we add $c_i \cdot K_{\psi_1}$ to each counter value at a position  
$\geq \beta(i,2)$) still results in a
$\sem{\psi_1}_{T(\O)}$-path by induction hypothesis, since we reduced the
height of $\pi$ by a multiple of $K_{\psi_1}$.
Our overall goal is to flatten $\pi$ by gluing together states only of certain 
intervals such that we obtain a path whose height is in total by precisely
$K\cdot K_{\psi_1}$ smaller than $\pi$'s.
Recall that there are $k\cdot K$ many intervals.
By the pigeonhole principle there is some $c\in[k]$ such that $c_i=c$ for at
least $K$ many intervals $I_i$. By gluing together $\frac{K}{c} \in \mathbb{N}$ pairs of states of distance
$c\cdot K_{\psi_1}$ each, we reduce $\pi$'s height by exactly
$\frac{K}{c}\cdot c\cdot K_{\psi_1}=K\cdot K_{\psi_1}$. This proves the claim.

Let us finish the proof the ``if''-direction.
Since by assumption $(q,n')\in\sem{\varphi}_{T(\O)}$, there exists
a finite path $ \pi\ =\ (q_1,n_1)\rightarrow(q_2,n_2)\to \cdots \to (q_l,n_l)$,
where $\pi[1,l-1]$ is a $\sem{\psi_1}_{T(\O)}$-path, $(q,n')=(q_1,n_1)$, 
and where $(q_{l},n_{l})\in\sem{\psi_2}_{T(\O)}$. 
To prove $(q,n)\in\sem{\varphi}_{T(\O)}$, we will assume that
$n_j>T$ for each $j\in[l]$. The case when $n_j=T$ for some $j\in[l]$ can
be proven similarly.
Assume first that the path $\pi[1,l-1]$ contains two states whose counter difference
is at least $k^2\cdot K\cdot K_{\psi_1}+K_\varphi$ which is (strictly) greater
than $k^2\cdot K\cdot K_{\psi_1}$. Since $K_\varphi$ is a multiple of $K\cdot
K_{\psi_1}$ by definition, we can apply the above claim $\frac{K_\varphi}{K\cdot
K_{\psi_1}}\in\N$ many times to $\pi[1,l-1]$. This reduces the height by
$K_\varphi$. We repeat this flattening process of $\pi[1,l-1]$ by height
$K_\varphi$ as long as possible, i.e., until any two states
have counter difference smaller than $k^2\cdot K \cdot K_{\psi_1}+K_\varphi$.
Let $\sigma$ denote the $\sem{\psi_1}_{T(\O)}$-path starting in $(q,n')$ that we obtain
from $\pi[1,l-1]$ by this process. Thus, $\sigma$ ends in some state, whose
counter value is congruent $n_{l-1}$ modulo $K_\varphi$ (since we
flattened $\pi[1,l-1]$ by a multiple of $K_\varphi$). 
Since $K_\varphi$ is in turn a multiple of $K_{\psi_2}$, we can
build a path $\sigma'$ which extends the path $\sigma$ by a single
transition to some state that satisfies $\psi_2$ by induction hypothesis. 
Moreover, by our flattening process,
the counter difference between any two states in $\sigma'$ is at most
$k^2\cdot K\cdot K_{\psi_1}+K_\varphi\leq 2\cdot k^2\cdot K_\varphi$. 
Recall that $T=\max\{2\cdot|\psi_i|\cdot k^2\cdot K_{\psi_i}\mid i\in\{1,2\}\}$.
As
$$n \ >\ 2\cdot|\varphi|\cdot k^2\cdot K_\varphi\ =\
2\cdot(|\varphi|-1+1)\cdot k^2\cdot K_\varphi\ \geq\ T+2\cdot k^2\cdot K_\varphi,$$ 
it follows that the path
that results from $\sigma'$ by subtracting $d$ from each counter value
(this path starts in $(q,n)$) is strictly above $T$. 
Moreover, since $d$ is a multiple of 
$K_{\psi_1}$ and $K_{\psi_2}$, this path
witnesses $(q,n)\in\sem{\varphi}_{T(\O)}$ by induction hypothesis.
\end{proof}
The following result can be obtained basically by 
using the standard model checking algorithm for $\CTL$ on finite
systems (see e.g. \cite{BaiKat08}) in combination with Thm.~\ref{thm-CTL-periodic}.

\begin{theorem} \label{CTL upper bound_0}
For a given one-counter process $\O=(Q,\{Q_p\mid p\in\Prop\},\delta_0,\delta_{>0})$,
a $\CTL$ formula $\varphi$, a control location $q\in Q$, and $n\in\N$ given in
binary, one can decide $(q,n)\in\sem{\varphi}_{T(\O)}$ in time 
$O(\log(n) +|Q|^3 \cdot |\varphi|^2 \cdot 4^{|Q| \cdot \lud(\varphi)} \cdot |\delta_0 \cup \delta_{>0}|)$.
\end{theorem}

As a corollary, we can deduce that for every fixed OCP $\O$ and
every fixed $k$ the question if for a given state $s$ and a given CTL formula
$\varphi$ with $\lud(\varphi)\leq k$, we have $(T(\O),s)\models\varphi$, is in
$\P$. This generalizes a result from \cite{GoMaTo09}, stating that the expression complexity of 
$\EF$ over OCPs is in $\mathsf{P}$.

\section{Expression complexity for CTL is hard for PSPACE}
{\label{S Expression}}

The goal of this section is to prove that model checking $\CTL$ is $\PSPACE$-hard
already over a fixed OCN.
We show this via a reduction from the well-known $\PSPACE$-complete problem QBF.
Our lower bound proof is separated into three steps. In step one, we 
define a family of $\CTL$ formulas $(\varphi_i)_{i\geq 1}$ such that over the
fixed OCN $\O$ that is depicted in Fig.~\ref{Fig-fixed-CCP} we can express (non-)divisibility by $2^i$.
In step two, we define a family of $\CTL$ formulas $(\psi_i)_{i\geq 1}$ such that
over $\O$ we can express if the $i^{\text{th}}$
bit in the binary representation of a natural is set to $1$.
In our final step, we give the reduction from QBF. 
For step one, we need the following simple fact which characterizes divisibility by
powers of two (recall that $[n] = \{1,\ldots,n\}$, in particular $[0] = \emptyset$):
\begin{equation}\label{fact1}
\forall n \geq 0, i \geq 1 : \text{ $2^i$ divides $n$ } \  \Leftrightarrow \
( 2^{i-1} \text{ divides } n  \;\wedge\; |\{n'\in[n] \mid 2^{i-1 }\text{ divides
} n' \}| \text{ is even})
\end{equation}
\begin{figure}[t]
\begin{center}
\setlength{\unitlength}{0.035cm}
\begin{picture}(375,100)(-65,20)
\put(-65,40){\huge$\delta_{>0}:$}
\gasset{Nframe=n,loopdiam=9,ELdist=.7}
\gasset{Nadjust=wh,Nadjustdist=1}
\gasset{curvedepth=0}
\node(tb)(25,80){$\overline{t}$}
\node(t)(25,40){$t$}
\node(q0)(80,20){$q_0$}
\node(q2)(80,100){$q_2$}
\node(q1)(115,60){$q_1$}
\drawloop[loopangle=0](q1){$-1$}	
\node(q3)(45,60){$q_3$}
\drawloop[ELpos=66,loopangle=240](q3){\tiny$-1$}	
\drawedge[ELside=r](q0,q1){\tiny$-1$}
\drawedge[ELside=r](q1,q2){\tiny$-1$}
\drawedge[ELside=r](q2,q3){\tiny$-1$}
\drawedge[ELside=l](q3,q0){\tiny$-1$}
\node(f)(-10,50){$f$}
\node(g)(-10,90){$g$}
\drawedge[AHnb=1,ATnb=1,ELside=l,curvedepth=0](q0,t){\tiny$0$}
\drawedge[AHnb=1,ATnb=1,ELside=l,ELpos=30,curvedepth=0](q1,tb){\tiny$0$}
\drawedge[AHnb=0,ATnb=1,ELside=r,curvedepth=0](q2,tb){\tiny$0$}
\drawedge[AHnb=1,ATnb=1,ELside=l,curvedepth=0](q3,tb){\tiny$0$}
\drawedge[AHnb=1,ELside=l,curvedepth=17](q2,t){\tiny$0$}
\drawedge[ELside=l,curvedepth=0](t,f){\tiny$0$}
\drawedge[ELside=l,curvedepth=0](tb,f){\tiny$-1$}
\drawedge[ELside=l,curvedepth=0,AHnb=1,ATnb=1](g,f){\tiny$-1$}
\node(p0)(5,120){$p_0$}
\node(p1)(45,120){$p_1$}
\drawedge[ELside=l,ELpos=50,curvedepth=5](tb,p1){\tiny$+1$}
\drawedge[ELside=l,ELpos=50,curvedepth=1](p1,tb){\tiny$0$}
\drawedge[ATnb=1,AHnb=1,ELside=r](p0,tb){\tiny$0$}
\drawloop[loopdiam=9,loopangle=0](p1){\tiny$+1$}

\put(165,40){\huge$\delta_{0}:$}
\gasset{Nframe=n,loopdiam=9,ELdist=.7}
\gasset{Nadjust=wh,Nadjustdist=1}
\gasset{curvedepth=0}
\node(tb)(225,25){$\overline{t}$}
\node(t)(285,35){$t$}
\node(q0)(310,25){$q_0$}
\node(f)(260,45){$f$}
\drawedge[ELside=r](t,q0){\tiny$0$}
\drawedge(t,f){\tiny$0$}
\node(p0)(205,65){$p_0$}
\node(p1)(245,65){$p_1$}
\drawedge[ELside=l,ELpos=50,curvedepth=0](tb,p1){\tiny$+1$}
\drawedge[ATnb=1,AHnb=1,ELside=r](p0,tb){\tiny$0$}	
\end{picture}
\end{center}
\caption{\label{Fig-fixed-CCP} The one-counter net $\O$ for which $\CTL$ model checking
is $\PSPACE$-hard}
\end{figure}
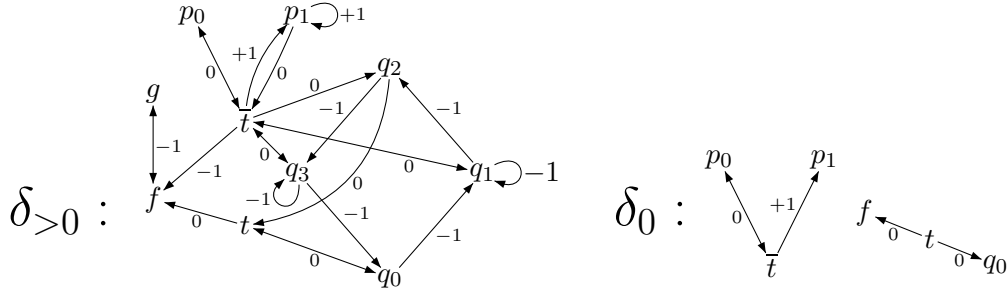
The set of propositions of $\O$ in Fig.~\ref{Fig-fixed-CCP}  
coincides with its control locations.
Recall that $\O$'s zero transitions are denoted by $\delta_0$ and $\O$'s positive
transitions are denoted by $\delta_{>0}$. Since $\delta_0\subseteq\delta_{>0}$, 
$\O$ is indeed an OCN.
Note that both $t$ and $\overline{t}$ are control locations of $\O$.
Now we define a family of $\CTL$ formulas
$(\varphi_i)_{i\geq 1}$ such
that for each $n\in\N$ we have: (i) $(t,n)\models\varphi_i$ if
and only if $2^i$ divides $n$ 
and (ii) $(\overline{t},n)\models\varphi_i$ if and only if 
$2^i$ does {\em not} divide $n$.
On first sight, it might seem superfluous to let the control location $t$
represent divisibility by powers of two and the control location $\overline{t}$ to represent
non-divisibility by powers of two since $\CTL$ allows negation. 
However the fact that we have {\em only one} family of formulas
$(\varphi_i)_{i\geq 1}$ to
express both divisibility and non-divisibility is a crucial technical subtlety
that is necessary in order to avoid an exponential blowup in formula size.
By making use of (\ref{fact1}), we construct 
the formulas $\varphi_i$ inductively.
\newcommand{\test}{\text{test}}
First, let us define the auxiliary formulas $\test=t\vee\overline{t}$ and
$\varphi_\diamond=q_0\vee q_1\vee q_2\vee q_3$.
Think of $\varphi_\diamond$ to hold in those control locations that altogether
are situated in the ``diamond'' in Fig.~\ref{Fig-fixed-CCP}.
We define
\begin{eqnarray*}
\varphi_1 &=& \test \wedge
\exists\X \, \left( f\wedge\EF(f\wedge\neg\exists\X g)\right) \text{ and } \\
\varphi_i & = & \test\ \wedge\ \exists\X \, \bigl( \exists (\varphi_\diamond\wedge
        \exists\X\varphi_{i-1})\ \U\ (q_0\wedge\neg\exists\X q_1) \bigr)
        \text{ for } i > 1.
\end{eqnarray*}
Since $\varphi_{i-1}$ is only used once in $\varphi_i$, we get
$|\varphi_i|\in O(i)$.
The following lemma states the correctness of the construction.

\begin{lemma}{\label{L Correctness}}
Let $n\geq 0$ and $i\geq 1$. Then 
\begin{itemize}
\item $(t,n)\models\varphi_i$ if and only if $2^i$ divides $n$.
\item $(\overline{t},n)\models\varphi_i$ if and only if $2^i$ does
not divide $n$.
\end{itemize}
\end{lemma}

\noindent
{\em Proof sketch.}
The lemma is proved by induction on $i$. The induction base
for $i = 1$ is easy to check. 
For $i > 1$, observe that $\varphi_i$ can only be true either in control location $t$ or
$\overline{t}$. 
Note that the formula right to the until symbol in $\varphi_i$ expresses that we are in $q_0$
and that the current counter value is zero.
Also note that the formula left to the until symbol
requires that $\varphi_\diamond$ holds, i.e., we are always in one of the four
``diamond control locations''. In other words, we decrement the counter
by moving along the diamond control locations (by possibly looping at $q_1$
and $q_3$) and always check if
$\exists\X\varphi_{i-1}$ holds, just until we are in $q_0$ and the counter value
is zero. Since there are transitions from $q_1$ and $q_3$ to $\overline{t}$
(but not to $t$), the induction hypothesis implies that 
the formula $\exists\X\varphi_{i-1}$ can be only true in $q_1$ and $q_3$ as
long as the current counter value is not divisible by $2^{i-1}$.
Similarly, since there are transitions from $q_0$ and $q_2$ to $t$
(but not to $\overline t$), the induction hypothesis implies that 
the formula $\exists\X\varphi_{i-1}$ can be only true in $q_0$ and $q_2$ if
the current counter value is divisible by $2^{i-1}$. With (\ref{fact1})
this implies the lemma.
\qed

\medskip

\noindent
For expressing if the $i^{\text{th}}$ bit of a natural is set to $1$, we make use of the
following simple fact:
\begin{equation}
\forall n \geq 0, i \geq 1 : 
\bit_i(n)=1 \ \Longleftrightarrow \
|\{n'\in[n]\mid 2^{i-1} \text{ divides } n'\}| \text{ is odd}
\end{equation}
Let us now define a family of $\CTL$ formulas $(\psi_i)_{i\geq 1}$
such that for each $n\in\N$ we have
$\bit_i(n)=1$	if and only if
$(\overline{t},n)\models\psi_i$. We set
$\psi_1= \varphi_1$ and 
$\psi_i=\overline{t} \wedge \exists\X \left((q_1\vee q_2)\ \wedge\ \mu_i
\right)$, 
where $\mu_i=
\exists (\varphi_\diamond\wedge
        \exists\X\varphi_{i-1})\ \U\ (q_0\wedge\neg\exists\X q_1)$
for each $i>1$.
Due to the construction of $\psi_i$ and since $|\varphi_i|\in O(i)$, we
obtain that $|\psi_i|\in O(i)$.
The following lemma states the correctness of the construction.

\begin{lemma}{\label{L Bit}}
Let $n \geq 0$ and let $i\geq 1$. Then
$(\overline{t},n)\models\psi_i$ if and only if $\bit_i(n)=1$.
\end{lemma}
Let us sketch the final step of the reduction from QBF. For this, let us assume
some quantified Boolean formula
$\alpha = Q_kx_k\, Q_{k-1}x_{k-1} \cdots Q_1x_1 : \beta(x_1,\ldots,x_k)$,
where $\beta$ is a Boolean formula over variables
$\{x_1,\ldots,x_k\}$ and $Q_i\in\{\exists,\forall\}$ is a quantifier for
each $i\in[k]$.
Think of each truth assignment
$\vartheta:\{x_1,\ldots,x_k\}\rightarrow\{0,1\}$ to correspond to the
natural number $n(\vartheta)\in[0,2^k-1]$, where
$\bit_i(n(\vartheta))=1$ if and only if $\vartheta(x_i)=1$, for each $i\in[k]$.
Let $\widehat{\beta}$ be the CTL formula that is obtained 
from $\beta$ by replacing each occurrence of
$x_i$ by $\psi_i$, which corresponds to applying Lemma \ref{L Bit}. It remains to describe how we
deal with quantification. 
Think of this as to consecutively incrementing the
counter from state $(\overline{t},0)$ as follows.
First, setting the variable $x_k$ to $1$ will correspond to adding $2^{k-1}$ to the
counter and getting to state $(\overline{t},2^{k-1})$. Setting $x_k$ to $0$ on the other
hand will correspond to adding $0$ to the counter and hence remaining in
state $(\overline{t},0)$. Next, setting $x_{k-1}$ to $1$ corresponds to adding to the
current counter value $2^{k-2}$, whereas setting $x_{k-1}$ to $0$ corresponds to
adding $0$, as expected. These incrementation steps can be  achieved
using the formulas $\varphi_i$ from Lemma~\ref{L Correctness}.
Finally, after setting variable $x_1$ either to $0$ or
$1$, we verify if the CTL formula $\widehat{\beta}$ holds.
Formally, let $\bigcirc_i=\wedge$ if $Q_i=\exists$ and $\bigcirc_i=\,\rightarrow$ if
$Q_i=\forall$ for each $i\in[k]$ 
(recall that $Q_k,\ldots,Q_1$ are the quantifiers of our quantified Boolean
formula $\alpha$). Let
$\theta_1 = Q_1\X \, ((p_0\vee p_1)\bigcirc_1 \exists\X\, \widehat{\beta})$ and for $i\in[2,k]$:
$$
\theta_i = Q_i\X\; 
\left((p_0\vee p_1)\bigcirc_i
\exists\left((p_0\vee
\exists\X\, (\overline{t} \wedge \varphi_{i-1}))\
\U\ (\overline{t}\wedge \neg\varphi_{i-1}\wedge\theta_{i-1})
)\biggl.\right)\right).
$$
Then, it can be show that
$\alpha$ is valid if and only if $(\overline{t},0)\in\sem{\theta_k}_{T(\O)}$.

\begin{theorem} \label{Theo-CTL-expression}
$\CTL$ model checking of the fixed OCN $\O$ from Fig.~\ref{Fig-fixed-CCP} is
$\PSPACE$-hard.
\end{theorem}
Note that the constructed $\CTL$ formula has leftward until
depth that depends on the size of $\alpha$.
By Thm.~\ref{CTL upper bound_0} this cannot be avoided unless $\P=\PSPACE$.
Observe that in order to express divisibility by powers of two, 
our $\CTL$ formulas $(\varphi_i)_{i\geq 0}$ 
have linearly growing leftward until depth.

\section{Tools from complexity theory}  \label{S Tools}

For Sec.~\ref{S Data} and \ref{further applications} we need some concepts
from complexity theory. By $\P^{\NP[\log]}$ we denote the class of all problems that can be solved on a
polynomially time bounded deterministic Turing machines which can have access to
an $\NP$-oracle only logarithmically many times, and by $\P^\NP$ the
corresponding class without the restriction to logarithmically many queries.
Let us briefly recall the definition of the circuit complexity class $\NC^1$,
more details can be found in \cite{Vol99}.
We consider Boolean circuits $C = C(x_1,\ldots,x_n)$ built up from 
AND- and OR-gates. Each input gate is labeled with a variable $x_i$
or a negated variable $\neg x_i$. The output gates are linearly
ordered. Such a circuit computes a function $f_C : \{0,1\}^n \to \{0,1\}^m$,
where $m$ is the number of output gates, in the obvious
way. The \emph{fan-in of a circuit} is the maximal number of
incoming wires of a gate in the circuit.
The \emph{depth of a circuit} is the number of gates along a longest path
from an input gate to an output gate.
A \emph{logspace-uniform $\NC^1$-circuit family} 
is a sequence $(C_n)_{n \geq 1}$ of Boolean circuits such 
that for some polynomial $p(n)$ and constant $c$:
(i) $C_n$ contains at most $p(n)$ many gates,
(ii) the depth of $C_n$ is at most $c \cdot \log(n)$,
(iii) the fan-in of $C_n$ is at most $2$, 
(iv) for each $m$ there is at most one circuit in $(C_n)_{n \geq 1}$
with exactly $m$ input gates, and
(v) there exists a logspace transducer
that computes on input $1^n$ a representation (e.g. as a 
node-labeled graph) of the circuit $C_n$.
Such a circuit family computes a 
partial mapping on $\{0,1\}^*$ in the 
obvious way (note that we do not require to have for every $n \geq 0$ 
a circuit with exactly $n$ input gates in the family, therefore
the computed mapping is in general only partially defined).
In the literature on circuit complexity one can find more restrictive notions
of uniformity, see e.g. \cite{Vol99}, but logspace uniformity 
suffices for our purposes. In fact, polynomial time uniformity
suffices for proving our lower bounds w.r.t. 
polynomial time reductions.

For $m \geq 1$ and $0 \leq M \leq 2^m-1$ 
let $\BIN_m(M) = \bit_m(M) \cdots \bit_1(M) \in \{0,1\}^m$ 
denote the $m$-bit binary representation of $M$.
Let $p_i$ denote the $i^{\text{th}}$ prime number.
It is well-known that the $i^{\text{th}}$ prime 
requires $O(\log(i))$ bits in its binary representation.
For a number $0 \leq M < \prod_{i=1}^m p_i$ we define the 
{\em Chinese remainder representation} $\CRR_m(M)$ as the 
Boolean tuple 
$\CRR_m(M) = (x_{i,r})_{i \in [m], 0 \leq r < p_i}$ with
$x_{i,r} = 1$ if $M \text{ mod } p_i = r$ and $x_{i,r}=0$ else.
By the following theorem, one can transform a Chinese remainder representation very
efficiently into binary representation.

\begin{theorem}[{\cite{ChDaLi01}}] \label{theorem hesse und co}
There is a logspace-uniform $\NC^1$-circuit family 
$(B_m( (x_{i,r})_{i \in [m], 0 \leq r < p_i} ))_{m \geq 1}$  such that 
for every $m \geq 1$, $B_m$ has $m$ output gates and
for every
$0 \leq M < \prod_{i=1}^m p_i$ we have that $B_m( \CRR_m(M) ) =
\BIN_m(M \text{ mod } 2^m)$.
\end{theorem}
By \cite{HeAlBa02}, we could replace logspace-uniform $\NC^1$-circuits
in Thm.~\ref{theorem hesse und co} even by $\mathsf{DLOGTIME}$-uniform
$\TC^0$-circuits. The existence of a $\mathsf{P}$-uniform $\NC^1$-circuit family
for converting from Chinese remainder representation to binary representation
was already shown in \cite{BCH86}.
Usually the Chinese remainder representation of $M$ is the tuple $(r_i)_{i \in [m]}$,
where $r_i = M  \text{ mod } p_i$. Since the primes $p_i$ will be always given in unary
notation, there is no essential difference between this representation and our 
Chinese remainder representation. The latter is more suitable for our purpose. 

The following definition of $\NC^1$-serializability is a variant 
of the more classical notion of serializability \cite{CaFu91,HLSVW93},
which fits our purpose better.
A language $L$ is {\em $\NC^1$-serializable}
if there exists an NFA $A$ over the alphabet
$\{0,1\}$, a polynomial $p(n)$, and a logspace-uniform
$\NC^1$-circuit family $(C_n)_{n \geq 0}$, where
$C_n$ has exactly $n+p(n)$ many inputs and one output, such that for every 
$x \in \{0,1\}^n$ we have 
$ x \in L$ if and only if $C_n(x,0^{p(n)}) \cdots C_n(x,1^{p(n)}) \in L(A)$,
where ``$\cdots$'' refers to the lexicographic order on $\{0,1\}^{p(n)}$.
With this definition, it can be shown that all languages in $\PSPACE$ are
$\NC^1$-serializable. A proof can be found in the appendix of \cite{GoLo09};
it is just a slight adaptation of the proofs from \cite{CaFu91,HLSVW93}.

\section{Data complexity for CTL is hard for PSPACE}  \label{S Data}

In this section, we prove that also the data complexity of $\CTL$ over
OCNs is hard for $\PSPACE$ and therefore $\PSPACE$-complete
by the known upper bounds for the modal $\mu$-calculus \cite{Serr06}.
Let us fix the set of propositions $\Pmc=\{\alpha,\beta,\gamma\}$ 
for this section.
In the following, w.l.o.g. we allow in $\delta_0$ (resp. in $\delta_{>0}$) 
transitions of the kind $(q,k,q')$, where $k\in\N$ (resp. $k\in\Z$) is given
in unary representation with the expected intuitive meaning. 

\begin{proposition} \label{prop main}
For the fixed $\EF$ formula $\varphi = (\alpha  \to \exists\mathsf{X} (\beta
\wedge \EF ( \neg \exists \mathsf{X} \gamma)))$ 
the following problem can be solved with a logspace transducer:

\noindent
INPUT: A list $p_1,\ldots,p_m$ of the first $m$ consecutive (unary encoded) prime numbers and
a Boolean formula $F = F((x_{i,r})_{i \in [m], 0 \leq r < p_i})$

\noindent
OUTPUT:  An OCN $\Omc(F)$ with distinguished control locations $\inp$ and $\out$,
such that for every number $0 \leq M < \prod_{i=1}^m p_i$ 
we have that $F(\CRR_m(M))=1$ if and only if 
there exists a $\sem{\varphi}_{T(\Omc(F))}$-path from $(\inp,M)$ to $(\out,M)$
in the transition system $T(\Omc(F))$.
\end{proposition}

\begin{proof}
W.l.o.g., negations occur in $F$ only
in front of variables. Then additionally, a negated variable $\neg x_{i,r}$ 
can be replaced by the disjunction $\bigvee \{ x_{i,k} \mid 0 \leq k < p_i, r
\neq k \}$. This can be done in logspace, since the primes $p_i$ 
are given in unary. Thus, we can assume that $F$ does not contain
negations.

The idea is to traverse the Boolean formula $F$
with the  OCN $\Omc(F)$ in a depth first manner. Each time
a variable $x_{i,r}$ is  seen, the OCN may also enter another branch, where it is 
checked, whether the current counter value is congruent 
$r$ modulo $p_i$. Let $\Omc(F) = (Q, \{Q_\alpha, Q_\beta, Q_\gamma\},
\delta_0, \delta_{>0})$, where
$Q = \{ \inp(G), \out(G) \mid G \text{ is a subformula of } F \} \cup 
\{ \div(p_1), \ldots,\div(p_m), \perp \}$, 
$Q_\alpha = \{ \inp(x_{i,r}) \mid i \in [m], 0 \leq r < p_i \}$,
$Q_\beta = \{ \div(p_1), \ldots,\div(p_m) \}$, and
$Q_\gamma =  \{ \perp \}$.
We set $\inp = \inp(F)$ and $\out = \out(F)$.
Let us now define the transition sets $\delta_0$ and $\delta_{>0}$.
For every subformula $G_1 \wedge G_2$ or $G_1 \vee G_2$ of $F$
we add the following transitions to $\delta_0$ and $\delta_{>0}$:
\begin{gather*}
\inp(G_1 \wedge G_2) \xrightarrow{0} \inp(G_1), \ \out(G_1) \xrightarrow{0} \inp(G_2),
\ \out(G_2) \xrightarrow{0} \out(G_1 \wedge G_2) \\
\inp(G_1 \vee G_2) \xrightarrow{0} \inp(G_i), \ \out(G_i) \xrightarrow{0}
\out(G_1 \vee G_2) \text{ for all } i \in \{1,2\}
\end{gather*}
For every variable $x_{i,r}$
we add to $\delta_0$ and $\delta_{>0}$
the transition
$\inp(x_{i,r}) \xrightarrow{0} \out(x_{i,r})$.
Moreover, we add to $\delta_{>0}$ the transitions
$\inp(x_{i,r})  \xrightarrow{-r} \div(p_i)$.
The transition $\inp(x_{i,0}) \xrightarrow{0} \div(p_i)$ 
is also added to $\delta_0$.
For the control locations $\div(p_i)$ we add to $\delta_{>0}$ 
the transitions  $\div(p_i) \xrightarrow{-p_i} \div(p_i)$ and
$\div(p_i) \xrightarrow{-1} \perp$.
This concludes the description of the OCN $\Omc(F)$.
Correctness of the construction can be easily checked by
induction on the structure of the formula $F$.
\end{proof}
We are now ready to prove $\PSPACE$-hardness of the data complexity.

\begin{theorem} \label{theo ctl data}
There exists a fixed $\CTL$ formula 
of the form $\exists \varphi_1 \U \varphi_2$, where
$\varphi_1$ and $\varphi_2$ are $\EF$ formulas, for which it 
is $\PSPACE$-complete to decide $(T(\Omc), (q,0)) \models \exists \varphi_1 \U
\varphi_2$ for a given OCN $\Omc$ and a control location $q$ of $\Omc$.
\end{theorem}

\begin{proof}
Let us take an arbitrary language $L$ in $\PSPACE$. 
Recall from Sec.~\ref{S Tools}
that $\PSPACE$ is $\NC^1$-serializable.
Thus, there exists an NFA $A=(S,\{0,1\},\delta,s_0,S_f)$ over the alphabet
$\{0,1\}$, a polynomial $p(n)$, and a logspace-uniform
$\NC^1$-circuit family $(C_n)_{n \geq 0}$, where
$C_n$ has $n+p(n)$ many inputs and one output, such that for every 
$x \in \{0,1\}^n$ we have:
\begin{equation} \label{eq C_n}
x \in L \ \Longleftrightarrow \ C_n(x,0^{p(n)}) \cdots C_n(x,1^{p(n)}) \in L(A),
\end{equation}
where ``$\cdots$'' refers to the lexicographic order on $\{0,1\}^{p(n)}$.
Fix an input $x \in \{0,1\}^n$. 
Our reduction can be split into the following five steps:

\medskip

\noindent
{\em Step 1.}
Construct in logspace the circuit $C_n$.
Fix the the first $n$ inputs of $C_n$ to the bits in $x$,
and denote the resulting circuit by $C$; it has only
$m = p(n)$ many inputs. Then, (\ref{eq C_n}) can 
be written as
\begin{equation} \label{eq C}
x \in L \ \Longleftrightarrow \ \prod_{M=0}^{2^m-1} C(\BIN_m(M)) \in L(A).
\end{equation}
{\em Step 2.} Compute the first $m$ consecutive primes
$p_1, \ldots, p_m$. This is possible in logspace, see e.g. 
\cite{ChDaLi01}. Every $p_i$ is bounded polynomially in $n$. 
Hence, every $p_i$ can be written down in unary notation.
Note that $\prod_{i=1}^m p_i > 2^m$ (if $m > 1$).

\medskip

\noindent
{\em Step 3.} Compute in logspace
the circuit $B = B_m((x_{i,r})_{i \in [m], 0 \leq r < p_i} )$
from Thm.~\ref{theorem hesse und co}.
Thus, $B$ is a Boolean circuit of fan-in 2 and
depth $O(\log(m)) = O(\log(n))$ with $m$ output gates and
$B(\CRR_m(M)) = \BIN_m(M \text{ mod } 2^m)$ 
for every $0 \leq M < \prod_{i=1}^m p_i$.

\medskip

\noindent
{\em Step 4.}
Now we compose the circuits $B$ and $C$: For every $i\in[m]$, connect the $i^{\text{th}}$ input of the circuit
$C(x_1,\ldots,x_m)$ with the $i^{\text{th}}$ output of the circuit $B$.
The result is a circuit with fan-in 2 and depth $O(\log(n))$.
In logspace, we can unfold this circuit into a Boolean formula 
$F = F((x_{i,r})_{i  \in [m], 0 \leq r < p_i})$. The resulting formula (or tree) has the same
depth as the circuit, i.e., depth $O(\log(n))$ and every tree node has at most
2 children. Hence, $F$ has polynomial size.
For every $0 \leq M < 2^m$ we have
$F(\CRR_m(M)) = C(\BIN_m(M))$ and equivalence (\ref{eq C}) can be written
as
\begin{equation} \label{gl F}
x \in L \ \Longleftrightarrow \ \prod_{M=0}^{2^m-1} F(\CRR_m(M)) \in L(A).
\end{equation}
{\em Step 5.} 
We now apply our construction from Prop.~\ref{prop main} to the formula $F$.
More precisely, let $G$ be the Boolean formula $\bigwedge_{i\in [m]} x_{i,r_i}$
where $r_i = 2^m \text{ mod } p_i$ for $i\in[m]$ (these remainders
can be computed in logspace).
For every $1$-labeled transition $\tau \in \delta$ of the NFA $A$ let $\Omc(\tau)$
be a copy of the OCN $\Omc(F \wedge \neg G)$. 
For every $0$-labeled transition $\tau \in \delta$ let
$\Omc(\tau)$ be a copy of the OCN $\Omc(\neg F\wedge \neg G)$.
In both cases we write $\Omc(\tau)$
as $(Q(\tau),\{Q_\alpha(\tau),Q_\beta(\tau),Q_\gamma(\tau)\}, \delta_0(\tau),\delta_{>0}(\tau))$.
Denote with $\inp(\tau)$ 
(resp. $\out(\tau)$) the control location of this copy that corresponds to
$\inp$ (resp. $\out$) in $\Omc(F)$. Hence, for every $b$-labeled transition
$\tau \in \delta$ ($b \in \{0,1\}$) and every $0 \leq M < \prod_{i=1}^m p_i$
there exists a $\sem{\varphi}_{T(\Omc(\tau))}$-path  
($\varphi$ is from Prop.~\ref{prop main}) from
$(\inp(\tau),M)$ to $(\out(\tau),M)$ 
if and only if $F(\CRR_m(M))=b$ and $M \neq 2^m$.

We now define an OCN $\Omc = (Q, \{Q_\alpha, Q_\beta,Q_\gamma\}, \delta_0, \delta_{>0})$ as follows:
We take the disjoint union of all the OCNs $\Omc(\tau)$ for $\tau \in \delta$.
Moreover, every state $s \in S$ of the NFA $A$ becomes
a control location of $\Omc$, i.e. 
$Q= S \cup \bigcup_{\tau \in \delta} Q(\tau)$ and
$Q_p= \bigcup_{\tau \in \delta} Q_p(\tau)$  for each $p \in
\{\alpha,\beta,\gamma\}$.
We add to $\delta_0$ and $\delta_{>0}$ for
every $\tau = (s,b,t) \in \delta$ the transitions
$s \xrightarrow{0} \inp(\tau)$ and $\out(\tau) \xrightarrow{1} t$.
Then, by Prop.~\ref{prop main} and (\ref{gl F}) we have
$x \in L$ if and only if there exists a 
$\sem{\varphi}_{T(\Omc)}$-path in $T(\Omc)$ from
$(s_0,0)$ to $(s,2^m)$ for some $s \in S_f$. 
Also note that there is no $\sem{\varphi}_{T(\Omc)}$-path in $T(\Omc)$
from $(s_0,0)$ to some configuration $(s,M)$ with $s \in S$ and 
$M > 2^m$. 
It remains to add to $\Omc$ some structure that enables
$\Omc$ to check that the counter has reached the value $2^m$.
For this, use again Prop.~\ref{prop main} to construct the OCN $\Omc(G)$ ($G$ is from above) 
and add it disjointly to $\Omc$. Moreover, add
to $\delta_{>0}$ and $\delta_0$ the transitions 
$s \xrightarrow{0} \inp$ for all $s \in S_f$, where $\inp$ is the $\inp$ control location  of $\Omc(G)$.
Finally, introduce a new proposition $\rho$ and set
$Q_\rho = \{\out\}$, where $\out$ is the $\out$ control location of $\Omc(G)$.
By putting $q=s_0$ we obtain:
$x \in L$ if and only if $(T(\Omc), (q,0)) \models \exists  (\varphi \ \mathsf{U} \ \rho)$,
where $\varphi$ is from Prop.~\ref{prop main}.
This concludes the proof of the theorem.
\end{proof}
By slightly modifying the proof of Thm.~\ref{theo ctl data}, one can also prove
that the fixed CTL formula can chosen to be of the form 
$\exists\G\psi$, where $\psi$ is an $\EF$ formula.

\section{Two further applications: EF and one-counter Markov decision processes} \label{further applications}

In this section, we present two further applications 
of Thm.~\ref{theorem hesse und co} to OCPs.
First, we state that the combined complexity for 
$\EF$ over OCNs is hard for $\mathsf{P}^{\mathsf{NP}}$. 
For formulas represented succinctly by directed acyclic graphs
this was already shown in \cite{GoMaTo09}. The point here is that we use
the standard tree representation for formulas. 

\begin{theorem}{\label{T EF}}
It is $\mathsf{P}^{\mathsf{NP}}$-hard (and hence
$\mathsf{P}^{\mathsf{NP}}$-complete by \cite{GoMaTo09})
to check $(T(\Omc), (q_0,0)) \models \varphi$ for given
OCN $\Omc$, state $q_0$ of $\Omc$, and $\EF$ formula $\varphi$.
\end{theorem}
The proof of Thm.~\ref{T EF} is very similar 
to the proof of Thm.~\ref{theo ctl data}, but does not use
the concept of serializability. We prove hardness by a reduction 
from the question whether the lexicographically maximal satisfying
assignment of a Boolean formula is even when interpreted as a natural number. This problem is 
$\mathsf{P}^{\mathsf{NP}}$-hard by \cite{Wag87}.
At the moment we cannot prove that the data complexity of 
$\EF$ over OCPs is hard for $\mathsf{P}^{\mathsf{NP}}$ 
(hardness for $\mathsf{P}^{\mathsf{NP}[\log]}$ was shown
in \cite{GoMaTo09}). Analyzing the proof of Thm.~\ref{T EF}
in \cite{GoLo09} shows that the main obstacle is the fact that 
converting from Chinese remainder representation into binary 
representation is not possible by uniform $\AC^0$ circuits
(polynomial size circuits of constant depth and unbounded
fan-in); this is provably the case.

In the rest of the paper, we sketch a second application of our lower bound
technique based on Thm.~\ref{theorem hesse und co}, see \cite{GoLo09} for more
details. This application
concerns  one-counter Markov decision processes.
{\em Markov decision processes} (MDPs) extend classical Markov chains by allowing so called
{\em nondeterministic vertices}. In these vertices,
no probability distribution on the outgoing
transitions is specified. The other vertices are called {\em probabilistic
vertices}; in these
vertices a probability distribution on the outgoing transitions is given.
The idea is that in an MDP a player Eve plays against nature (represented
by the probabilistic vertices).
In each nondeterministic vertex $v$, Eve chooses a probability distribution on the
outgoing transitions of $v$; this choice may depend on the past of the play (which is a path 
in the underlying graph ending in $v$) and is formally represented by a strategy for Eve.
An MDP together with a strategy for Eve 
defines a Markov chain, whose state space
is the unfolding of the graph underlying the MDP. Here, we consider infinite MDPs, which
are finitely represented by OCPs; this formalism was introduced in 
\cite{BraBroEteKucWoj09} under the name 
{\em one-counter Markov decision process} (OC-MDP).
With a given OC-MDP $\mathcal{A}$ and a set $R$ of control locations of the OCP underlying $\mathcal{A}$ 
(a so called {\em reachability constraint}),
two sets were associated in \cite{BraBroEteKucWoj09}:
$\ValOne(R)$ is the set of all vertices $s$
of the MDP defined by $\mathcal{A}$ such that for every $\epsilon > 0$ there exists
a strategy $\sigma$ for Eve under which the
probability of finally reaching from $s$ a control location in $R$ and at
the same time having counter value $0$ is at least $1-\varepsilon$.
$\OptValOne(R)$ is the set of all vertices $s$
of the MDP defined by $\mathcal{A}$ for which there exists a specific strategy
for Eve under which this probability is $1$. It was shown in \cite{BraBroEteKucWoj09}
that for a given OC-MDP $\mathcal{A}$, a set of control locations $R$, and a vertex $s$
of the MDP defined by $\mathcal{A}$, the question if $s \in \OptValOne(R)$  is  
$\PSPACE$-hard and in $\EXPTIME$. The same question for $\ValOne(R)$
instead of $\OptValOne(R)$ was shown to be hard 
for each level of the Boolean hierarchy $\BH$, which
is a hierarchy of complexity classes between $\NP$ and  $\mathsf{P}^{\mathsf{NP}[\log]}$.
By applying our lower bound techniques (from Thm.~\ref{theo ctl
data}) we can prove the following.
\begin{theorem}\label{thm markov}
Membership in $\ValOne(R)$ is $\PSPACE$-hard.
\end{theorem}

As a byproduct of our proof, we also reprove $\PSPACE$-hardness for $\OptValOne(R)$.
It is open, whether $\ValOne(R)$ is decidable; the corresponding problem for 
MDPs defined by pushdown processes is undecidable \cite{EtYa05}.

%% in general the use of bibtex is encouraged

\bibliographystyle{abbrv}

\def\cprime{$'$}

\end{document}